\documentclass[Royal,times]{sagej_NoJournalName}
\usepackage{moreverb,url}

\usepackage[colorlinks,bookmarksopen,bookmarksnumbered,citecolor=red,urlcolor=red]{hyperref}
\usepackage{amsmath}
\usepackage{float}
\usepackage{booktabs}
\usepackage{array}
\newcommand\BibTeX{{\rmfamily B\kern-.05em \textsc{i\kern-.025em b}\kern-.08em
T\kern-.1667em\lower.7ex\hbox{E}\kern-.125emX}}

\newtheorem{theorem}{Theorem}[section]

\newtheorem{proposition}[theorem]{Proposition}

\theoremstyle{definition}
\newtheorem{definition}[theorem]{Definition}

\theoremstyle{remark}

\begin{document}

\runninghead{Market Makers and Risk Aversion}

\title{Market Makers and Risk Aversion: A Hamiltonian Approach}

\author{Will Hicks\affilnum{1}}

\affiliation{\affilnum{1} Centre for Quantum Social and Cognitive Science, Memorial University of Newfoundland}

\corrauth{Will Hicks}

\email{whicks7940@googlemail.com}

\begin{abstract}
In this article we assess the extent to which unpredictable price changes in the financial market can arise from nonlinear coupling between the size of the market makers' inventory and the traded price. By encoding simple market forces into the dynamics we find that the market makers' risk appetite becomes the key parameter that determines the degree to which price changes can be reliably forecast in advance. We go onto analyse the time-series generated by using different sampling intervals and find that at certain intermediate intervals, the series generated cannot easily be distinguished from a GARCH process. At sufficiently long sampling intervals we are unable to reject the null hypotheses that the process is generated by independent and identically distributed returns. The article demonstrates that whilst external shocks and random noise are important in the treatment of financial time-series, they are not necessary in order to generate unpredictable price changes.
\end{abstract}
\maketitle
\section{Introduction}
It is a very well known phenomenon that financial market time-series show unpredictable returns. Indeed, were this not the case the process of making money in the stock market would be much easier than it is. Often in finance, unpredictability is introduced through new information:
\begin{itemize}
\item A random walk: the price jumps up (good news) or down (bad news).
\item A limit order book model: new buy/sell orders arrive according to a stochastic process. For example a Poisson process.
\end{itemize}
The objective of this article is to design simple models by which unpredictable price changes are the result of predictable forces acting on a market price. We don't aim to suggest that seemingly random price shocks from outside the market are not a key factor in modelling. The aim in this article is simply to assess the degree to which unpredictable price changes can occur without these factors.
\subsection{The Vasicek Model:}
Many models in finance and economics borrow intuitive ideas from the world of physics. The concept of ``equilibrium'' is just such an idea. To illustrate, consider the Vasicek model (see \cite{Vasicek}). In this model, the dyanamic variable $r_t$ (in this case the instantaneous short rate) experiences a mean reversion to some kind of equilibrium: $\overline{r}$, at a rate determined by the mean reversion speed: $\gamma$. In addition, one accounts for the unpredictable element of time evolution using a stochastic process $W_t$:
\begin{align}\label{vasicek}
dr_t&=\gamma(\overline{r}-r_t)dt+\sigma dW_t
\end{align}
In fact, one can write down the time evolution as a stochastic ordinary differential equation:
\begin{align*}
M\frac{d^2r(t)}{dt^2}+\frac{dr(t)}{dt}+\gamma r(t)&=\sigma\xi(t)+\gamma\overline{r}\\
\xi(t)&\text{: a white noise}
\end{align*}
Where taking the limit: $M\rightarrow 0$, leads to the Vasicek model. That is we can instinctively understand the model as comprising:
\begin{enumerate}
\item[1)] Internal market forces comprising a linear restoring force to the equilibrium $\overline{r}$.
\item[2)] External, and unpredictable, events that impact the price.
\end{enumerate}
Although not directly addressed in the original article, one way of thinking about this approach might be that, in the absence of any random noise from outside the market, a high market price for a traded financial asset (one far above the equilibrium value) would lead to selling. A low price would lead to buying. This would lead to oscillations about the long term equilibrium.

The subsequent development of the Vasicek model centred around enriching the structure firstly to allow the model to be calibrated against the full bond curve (ie traded bonds with different maturities) and secondly to allow calibration against different derivative contracts that depend on the volatility parameter. In this article, for reasons discussed below, we ignore the stochastic element of the price evolution. Instead we focus on the internal market forces. We show how unpredictable price changes can occur when making relatively simple extensions to the internal market forces considered in the Vasicek approach.
\subsection{The Excess Volatility Puzzle:}
The excess volatility puzzle refers to the observation that financial prices seem to move around far more than can plausibly be explained by news events. In \cite{WMS} and later also in \cite{WMS2}, the authors look at the largest daily changes in equity markets and attempt to match the prices against external events. For example market crashes that occurred in October 2008 could easily  be linked to the announcement that Lehman Brothers (a large US bank) had filed for bankruptcy. In fact in \cite{WMS}, published in 1989, the authors found that the majority of the largest daily moves in major indices had no obvious cause. The authors of a follow up study in 2013 (see \cite{WMS2}) came to similar conclusions. In \cite{Bouchaud}, the authors apply the market microstructure approach to this problem. In this article we look at extending the simple assumptions underlying the Ornstein-Uhlenbeck process and the Vasicek model. The aim is to understand the impact on financial market dynamics, and to understand how turbulent price dynamics might arise from inside the system itself, rather than from external shocks. 
\subsection{The Financial Market as a Hamiltonian System:}
When one models the time evolution of a dynamic variable: $X$ (where $X$ can be a scalar or a vector) using a system of equations:
\begin{align*}
\frac{dX}{dt}=f(X)
\end{align*}
one may describe such systems as ``closed'' from a mathematical perspective in the sense that knowledge of $X$ is sufficient to determine the time evolution of $X$. There is no interaction with an external environment or any random element (such as a stochastic process: $W_t$). Despite this there is no reason to assume that such a system can be described by a Hamiltonian function. Often, one can view such systems as being dissipative. A typical example is the damped pendulum. However, this system is not closed in a physical sense, and energy is lost as heat into the environment. Therefore we make the assumption that any system that is truly ``closed'' in the the physical sense, must have some kind of conserved quantity, and thus be, in some way governed by a Hamiltonian function.

Ultimately, uncertain events and uncertain interaction with the world outside {\em will} impact what is going on inside the financial market. In any realistic model, these factors should be incorporated. Nonetheless, the objective in this article is to investigate the extent to which internal dynamics of the market can lead to unpredictable price changes, and whether this might be part of the answer to the excess volatility puzzle.

We note also that in the Vasicek case, the model takes the limit of mass to zero, and incorporates damping. The deterministic dynamics are dissipative, and the variance, after adding the noise that drives the system, is bounded above. Without the noise element, the system simply reverts instantaneously to the equilibrium point. In this article, the objective is to investigate the extent to which the internal dynamics of the market can drive complex and unpredicable price changes. Therefore, we take the view that a non-zero `mass' parameter is necessary to meaningfully investigate these effects.
\section{Market Makers and Static Risk Aversion}\label{static section}
\subsection{Hamiltonian Formulation:}
We start with the concept of a linear restoring force pulling the traded price $x$ back to the equilibrium price $x_0$:
\begin{align}\label{H_start}
H(x,P_x)&=\frac{P_x^2}{2M_x}+\frac{K_x(x-x_0)^2}{2}
\end{align}
We now imagine the roll of market makers in the setup, and introduce as a second variable, the inventory currently held by market makers: $v$. In general, market makers look to profit from order flow, rather than through taking risk. Therefore, we make the assumption that the inventory held by market makers experiences a linear restoring force to zero. In fact, the magnitude of the risk position for market makers is given by the inventory multiplied by the price: $xv$. Therefore, the Hamiltonian \ref{H_start} becomes:
\begin{align*}
H(x,v,P_x,P_v)&=\frac{P_x^2}{2M_x}+\frac{P_v^2}{2M_v}+\frac{K_x(x-x_0)^2}{2}+\frac{\epsilon(xv)^2}{2}
\end{align*}
Finally, we assume the inventory $v$ is itself constrained somehow by a potential function: $f(v)$ (for example where the total supply of the asset is limited), so that we have:
\begin{align}\label{H_static}
H(x,v,P_x,P_v)&=\frac{P_x^2}{2M_x}+\frac{P_v^2}{2M_v}+\frac{K_x(x-x_0)^2}{2}+\frac{\epsilon(xv)^2}{2}+f(v)
\end{align}
Note that if we simplify \ref{H_static} by using a quadratic potential (with linear restoring force) for $v$, then we find that the risk aversion term: $\frac{\epsilon(xv)^2}{2}$, represents a nonlinear perturbation of an otherwise integrable system. For this reason, we use the suggestive term: $\epsilon$ to represent the strength of the risk aversion.
\subsection{Canonical Perturbation Theory:}
We start by assuming a simple form for $f(v)$ and re-writing Hamiltonian \ref{H_static} using action-angle variables.
\begin{proposition}\label{prop_sys1_H}
Let the potential function for $v$ be given by: $f(v)=\frac{K_vv^2}{2}$. Then, the system described in equation \ref{H_static} can be written:
\begin{align}\label{H_simp1}
H(I_x,\theta_x,I_v,\theta_v)&=\omega_xI_x+\tilde{\omega}_vI_v+\epsilon\Big(\frac{2I_xI_v}{(K_xM_x\tilde{K}_vM_v)^{1/2}}\Big)\sin^2(\theta_x)\sin^2(\theta_v)\\
&+\epsilon\Big(\frac{(2\sqrt{2})x_0I_x^{1/2}I_v}{(K_xM_x)^{1/4}(\tilde{K}_vM_v)^{1/2}}\Big)\sin(\theta_x)\sin^2(\theta_v)\nonumber\\
\omega_x&=\sqrt{\frac{K_x}{M_x}}\text{, }\tilde{\omega}_v=\sqrt{\frac{\tilde{K}_v}{M_v}}\text{, }\tilde{K}_v=K_v+\epsilon x_0^2\nonumber
\end{align}
Where we have:
\begin{align*}
x&=\sqrt{\frac{2I_x}{M_x\omega_x}}\sin(\theta_x)+x_0\text{, }v=\sqrt{\frac{2I_v}{M_v\tilde{\omega}_v}}\sin(\theta_v)\\
p_x&=\sqrt{2I_xM_x\omega_x}\cos(\theta_x)\text{, }p_v=\sqrt{2I_vM_v\tilde{\omega}_v}\cos(\theta_v)
\end{align*}
\end{proposition}
\begin{proof}
Plugging the chosen action-angle variables $I_x,I_v$ and $\theta_x,\theta_v$ into equation \ref{H_static} gives back the original Hamiltonian \ref{H_simp1}.
\end{proof}
\subsection{Kolmogorov Arnold Moser Theory:}\label{KAM}
Note first that the Hamiltonian \ref{H_simp1} is of the form:
\begin{align*}
H(I_x,I_v,\theta_x,\theta_v)&=\omega_xI_x+\tilde{\omega}_vI_v+\epsilon H_1(I_x,I_v,\theta_x,\theta_v)
\end{align*}
If we set the risk appetite $\epsilon=0$, then the Hamiltonian takes the form of a harmonic oscillator. The system is integrable, and the solutions periodic:
\begin{align}\label{KAM_periodic}
\dot{I}_x&=\dot{I}_v=0\\
\dot{\theta}_x&=\omega_x\text{, }\dot{\theta}_v=\tilde{\omega}_v=\omega_v\nonumber
\end{align}
The price and market maker inventory simply oscillate about their respective equilibria. The objective now is to apply a coordinate transformation so that we can write the Hamiltonian in the form:
\begin{align*}
H(I_x,I_v,\theta_x,\theta_v)&=\omega_xI_x+\tilde{\omega}_vI_v+H'_1(I_x,I_v)+O(\epsilon^2)
\end{align*}
In fact, assuming the perturbation is sufficiently small we can transform the variables an arbitrary number of times to get:
\begin{align*}
H(I_x,I_v,\theta_x,\theta_v)&=\sum_{i=1}^{n-1}H_i(I_x,I_v)+O(\epsilon^{2n})
\end{align*}
That is, the simple periodic dynamics for the price and market maker inventory is retained. Following the standard method (see for example \cite{Arnold}) we have:
\begin{proposition}\label{KAM_Prop}
For a sufficiently small perturbation: $\epsilon$, the Hamiltonian \ref{H_simp1} can be written:
\begin{align}\label{H_static_KAM}
H(I_x,I_v,\theta_x,\theta_v)&=\omega_xI_x+\tilde{\omega}_vI_v+\epsilon\Big(\frac{I_xI_v}{2(M_xK_xM_v\tilde{K}_v)^{1/2}}\Big)+O(\epsilon^2)
\end{align}
For these small perturbations, the system retains its' periodic dynamics. That is the price and market marker inventory oscillate about their respective equilibria. The size of the perturbation required for the breakdown of the approximation, and the breakdown of periodic dynamics, is dependent on avoiding certain resonant frequencies. For example:
\begin{align*}
\omega_x\approx\tilde{\omega}_v\text{, and: }\omega_x\approx 2\tilde{\omega}_v
\end{align*}
\end{proposition}
\begin{proof}
We write Hamiltonian \ref{H_simp1} as:
\begin{align*}
H(I_x,I_v,\theta_x,\theta_v)&=\omega_xI_x+\tilde{\omega}_vI_v+\epsilon f_1(I_x,I_v,\theta_x,\theta_v)+\epsilon f_2(I_x,I_v,\theta_x,\theta_v)\\
f_1(I_x,I_v,\theta_x,\theta_v)&=\Big(\frac{2I_xI_v}{(M_XK_xM_v\tilde{K}_v})^{1/2}\Big)\sin^2(\theta_x)\sin^2(\theta_v)\\
f_2(I_x,I_v,\theta_x,\theta_v)&=\Big(\frac{(2\sqrt{2})x_0I_x^{1/2}I_v}{(K_xM_x)^{1/4}(\tilde{K}_vM_v)^{1/2}}\Big)\sin(\theta_x)\sin^2(\theta_v)
\end{align*}
We apply a type II generating function:
\begin{align*}
S(I_x',I_v',\theta_x,\theta_v)&=I_x'\theta_x+I_v'\theta_v+\epsilon G(I_x',I_v',\theta_x,\theta_v)
\end{align*}
For $G(I_x',I_v',\theta_x,\theta_v)$ yet to be determined. Therefore we have the transformation:
\begin{align*}
I_x&=I_x'+\epsilon\partial_{\theta_x} G\text{, }I_v=I_v'+\epsilon\partial_{\theta_v} G\\
\theta_x'&=\theta_x+\epsilon\partial_{I_x'} G\text{, }\theta_v'=\theta_v+\epsilon\partial_{I_v'} G
\end{align*}
We now have:
\begin{align*}
H(I_x',I_v',\theta_x',\theta_v')&=\omega_xI_x'+\tilde\omega_vI_v'+\epsilon\Big(\omega_x\partial_{\theta_x}G+\tilde\omega_v\partial_{\theta_v}G+f_1(I_x',I_v',\theta_x',\theta_v')\\
&+f_2(I_x',I_v',\theta_x',\theta_v')\Big)+O(\epsilon^2)
\end{align*}
Thus, if we can find:
\begin{align}\label{G_pde1}
\omega_x\partial_{\theta_x}G+\tilde\omega_v\partial_{\theta_v}G&=\langle f_1+f_2\rangle(I_x',I_v')-f_1(I_x'I_v',\theta_x',\theta_v')-f_2(I_x'I_v',\theta_x',\theta_v')\\
\langle f_1+f_2\rangle(I_x',I_v')&=\frac{1}{4\pi^2}\int_0^{2\pi}\int_0^{2\pi}f_1(I_x'I_v',\theta_x',\theta_v')+f_2(I_x'I_v',\theta_x',\theta_v')d\theta_x'd\theta_v'\nonumber
\end{align}
Then we have a `fast angular rotation' approximation:
\begin{align*}
H(I_x',I_v',\theta_x',\theta_v')=\omega_xI_x'+\tilde{\omega}_vI_v'+\langle f_1+f_2\rangle(I_x',I_v')
\end{align*}
Therefore, if we can solve equation \ref{G_pde1} to find an appropriate $G(I_x',I_v',\theta_x,\theta_v)$, then the result follows from the fact that in this case we have:
\begin{align*}
\langle f_1+f_2\rangle(I_x',I_v')&=\frac{I_x'I_v'}{2(M_xK_xM_v\tilde{K}_v)^{1/2}}
\end{align*}
Using the chain rule, \ref{G_pde1} can be written as a simple 1st order equation:
\begin{align}\label{G_pde2}
(\omega_x^2+\tilde\omega_v^2)&\,\partial_u G
=\frac{I_x' I_v'}{2\sqrt{K_xM_x\tilde K_vM_v}}
\Bigg[\cos\!\Big(\frac{2\omega_x u+2\tilde\omega_v w}{\omega_x^2+\tilde\omega_v^2}\Big)
     +\cos\!\Big(\frac{2\tilde\omega_v u-2\omega_x w}{\omega_x^2+\tilde\omega_v^2}\Big) \notag\\
     &-\frac12\cos\!\Big(\frac{2(\omega_x-\tilde\omega_v)u+2(\omega_x+\tilde\omega_v)w}{\omega_x^2+\tilde\omega_v^2}\Big)-\frac12\cos\!\Big(\frac{2(\omega_x+\tilde\omega_v)u-2(\omega_x-\tilde\omega_v)w}{\omega_x^2+\tilde\omega_v^2}\Big)\Bigg] \notag\\
&+\frac{2\sqrt2\,x_0\,(I_x')^{1/2}I_v'}{(K_xM_x)^{1/4}(\tilde K_vM_v)^{1/2}}
\Bigg[-\frac12\sin\!\Big(\frac{\omega_x u+\tilde\omega_v w}{\omega_x^2+\tilde\omega_v^2}\Big) \notag\\
     &+\frac14\sin\!\Big(\frac{(\omega_x+2\tilde\omega_v)u+(\tilde\omega_v-2\omega_x)w}{\omega_x^2+\tilde\omega_v^2}\Big)
     +\frac14\sin\!\Big(\frac{(\omega_x-2\tilde\omega_v)u+(\tilde\omega_v+2\omega_x)w}{\omega_x^2+\tilde\omega_v^2}\Big)\Bigg]
\end{align}
Where:
\begin{align*}
\begin{pmatrix}u\\w\end{pmatrix}&=\begin{pmatrix}\omega_x&&\tilde{\omega}_v\\\tilde{\omega}_v&&-\omega_x\end{pmatrix}\begin{pmatrix}\theta_x\\\theta_v\end{pmatrix}
\end{align*}
Thus we have:
\begin{align}\label{G_sol}
G ={}& \frac{I_x' I_v'}{2\sqrt{K_xM_x\tilde K_vM_v}}\Bigg[\frac{\sin 2\theta_x}{2\omega_x}+\frac{\sin 2\theta_v}{2\tilde\omega_v}
  -\frac{\sin(2\theta_x+2\theta_v)}{4(\omega_x+\tilde\omega_v)}-\frac{\sin(2\theta_x-2\theta_v)}{4(\omega_x-\tilde\omega_v)}\Bigg] \notag\\
&+\frac{2\sqrt2\,x_0\,(I_x')^{1/2}I_v'}{(K_xM_x)^{1/4}(\tilde K_vM_v)^{1/2}}
\Bigg[\frac{\cos\theta_x}{2\omega_x}-\frac{\cos(\theta_x+2\theta_v)}{4(\omega_x+2\tilde\omega_v)}-\frac{\cos(\theta_x-2\theta_v)}{4(\omega_x-2\tilde\omega_v)}\Bigg]
\end{align}
and equation \ref{G_pde2} is easily solvable as long as $|\omega_x-\tilde{\omega}_v|>>\epsilon$ and $|\omega_x-2\tilde{\omega}_v|>>\epsilon$.
\end{proof}
\subsection{Breakdown of Periodic Motion and the Onset of Chaos}
In fact, in order to prove that the dynamics are periodic for a particular $\epsilon$, requires us to ensure the expansion, where the first order term is given in proposition \ref{KAM_Prop}, converges as you repeat the process for higher and higher powers of $\epsilon$. In fact the $O(\epsilon^2)$ terms in \ref{H_static_KAM} will contain $\partial_{I_x'}G$ and $\partial_{I_v'}G$. From equation \ref{G_sol}, it can be seen that these will include terms like:
\begin{align*}
\frac{1}{\omega_x-\tilde\omega_v}\text{, and }\frac{1}{\omega_x-2\tilde\omega_v}
\end{align*}
So that if $|\omega_x-\tilde\omega_v|\sim\epsilon$, or $|\omega_x-2\tilde\omega_v|\sim\epsilon$, then the remaining terms are no longer $O(\epsilon^2)$. Looking back at the original Hamiltonian \ref{H_static}, one can see this as oscillations in the price causing resonance in the oscillations in the market maker inventory, and vice versa. As we expand to powers of $\epsilon^{2n}$, we will find other resonant frequencies that cause the series to diverge. However, from the analysis presented, we can conclude that the breakdown down of regular motion and the onset of chaos is driven by:
\begin{itemize}
\item Lower levels of risk appetite are consistent with higher levels of $\epsilon$. That is, for a system with a defined level of energy, the maximum potential is reached at lower values for $x,v$ with higher $\epsilon$. In this setup a more risk averse market is more likely to see chaotic motion. Market makers cannot warehouse sufficient risk (due to their risk aversion) to ensure the price returns to equilibrium in an orderly fashion.
\item For markets with higher risk appetite, market makers can warehouse a higher inventory, allowing them to more easily ensure the price reverts to the equilibrium in a predictable fashion.
\item Where the natural frequencies are close to certain resonant frequencies, for example $\omega_x\approx\tilde\omega_v$, $\omega_x\approx 2\tilde\omega_v$, the market is more prone to chaos.
\item Since we have: $\tilde{K}_v=K_v+\epsilon x_0^2$ increasing the risk aversion changes the resonant frequencies, as well as increasing the strength of the coupling.
\end{itemize}
\subsection{Lagrangian Approach:}
In this section we consider the Euler-Lagrange equations for the system.
\begin{proposition}{Static Lagrangian}\newline
For the Hamiltonian given by equation \ref{H_static}, the Euler-Lagrange equations become:
\begin{align}\label{EL_simp1}
\ddot{x}&+\bigg(\frac{K_x+\epsilon v^2}{M_x}\bigg)x=\frac{K_x}{M_x}x_0\\
\ddot{v}&+\frac{\epsilon x^2v+f'(v)}{M_v}=0\nonumber
\end{align}
\end{proposition}
\begin{proof}
Under Hamiltonian \ref{H_static}, the Lagrangian is given by:
\begin{align*}
L(x,\dot{x},v,\dot{v})&=\frac{M_x\dot{x}^2}{2}+\frac{M_v\dot{v}^2}{2}-\frac{K_x(x-x_0)^2}{2}-\frac{\epsilon(xv)^2}{2}-f(v)
\end{align*}
The canonical momentum is now given by:
\begin{align*}
\frac{\partial L}{\partial\dot{x}}&=M_x\dot{x}\rightarrow\frac{d}{dt}\bigg(\frac{\partial L}{\partial\dot{x}}\bigg)=M_x\ddot{x}\\
\frac{\partial L}{\partial\dot{v}}&=M_v\dot{v}\rightarrow\frac{d}{dt}\bigg(\frac{\partial L}{\partial\dot{v}}\bigg)=M_v\ddot{v}
\end{align*}
Further we have:
\begin{align*}
\frac{\partial L}{\partial x}&=-K_x(x-x_0)-\epsilon v^2x\\
\frac{\partial L}{\partial v}&=-f'(v)-\epsilon x^2v
\end{align*}
\end{proof}
\subsection{System Dyanamics:}
\begin{itemize}
\item We see from equation \ref{EL_simp1} that this system can be described as an anharmonic oscillator whereby an increase in the market maker inventory, together with non-zero risk aversion, increases the market `stiffness':\newline
$K_x\rightarrow K_x+\epsilon v^2$
\item Non-zero risk aversion also reduces the systems' equilibrium price:\newline
$x_0\rightarrow \bigg(\frac{K_x}{K_x+\epsilon v^2}\bigg)x_0$.
\item If we set $f(v)=\frac{K_vv^2}{2}$ then the market maker inventory takes the same form:\newline
$\ddot{v}+\bigg(\frac{K_v+\epsilon x^2}{M_v}\bigg)v=0$
\item An increase in the magnitude of the price increases the   stiffness whereby $v$ is forced back to zero. That is, the market maker's incentive to offload inventory is increased.
\item Finally, it is instructive to review the system potential function under increasing levels of risk aversion. This is shown in figure \ref{Pot}. Under zero risk aversion the potential is a standard harmonic potential well. As you increase the potential well tightens as risk appetite is fulfilled at lower levels in the price. However, it does so in such a way that the potential function becomes non-convex.
\item Figure \ref{Pot} illustrates the potential function for the case where the market maker inventory has a quadratic potential, with the same fundamental frequency as the market price.
\end{itemize}
\begin{figure}
\includegraphics[scale=0.7]{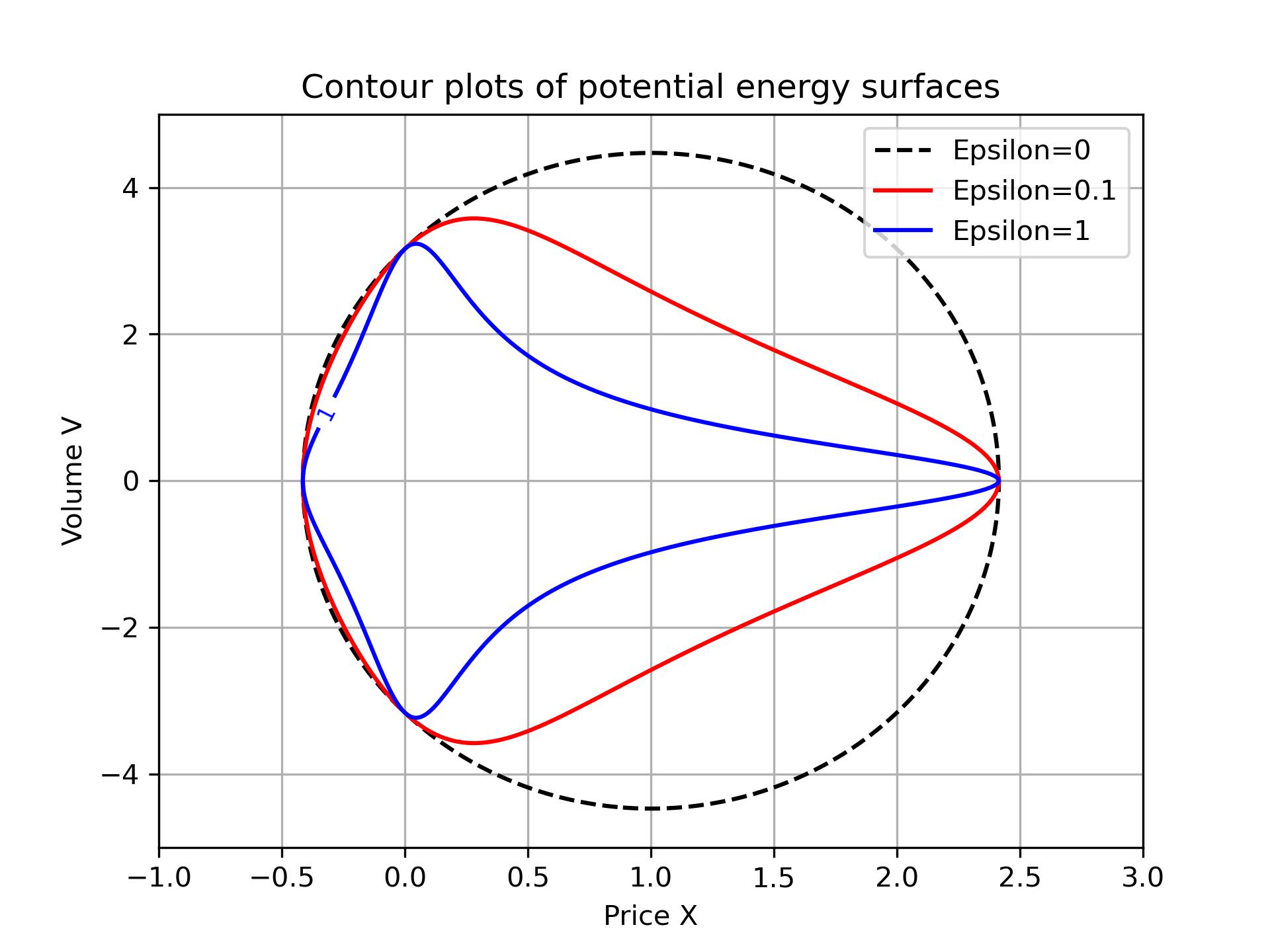}
\caption{Contour plots for the potential functions with different values for $\epsilon$. Each case has the same market energy. We set $x_0=1$.}\label{Pot}
\end{figure}
\section{Numerical Simulation:}
\subsection{Poincar\'{e} Section}
In this section we illustrate some of the key properties of the model by carrying out a numerical simulation of the system described by equation \ref{H_static}. This has been carried out using the pyHamSys Python library (see \cite{pyHamsys}, \cite{pyHamsys2}). First we illustrate the Poincar\'{e} section as follows:
\begin{itemize}
\item We set: $x_0=3$, $K_x=0.11,K_v=0.1,M_x=M_v=1$ throughout this section.
\item We look at variable values for $\epsilon$ and variable values for the system energy.
\item We sample random initial conditions, until the initial energy falls within $\pm 0.01$ of the energy target. This represents a single `path' for the simulation.
\item The Poincar\'{e} sections shown below are simulated using 100 such paths. We show the cross section for up crossings ($\dot{v}>0$) at $v=0$.
\end{itemize}
Figure \ref{PC_E_1} shows the Poincar\'{e} sections for the energy target set to $1$, and $\epsilon=0.0001,0.001,0.01,0.1$. At the lowest value for $\epsilon$, the motion closely resembles that for the harmonic oscillator. For the section with $\epsilon=0.001$, close inspection reveals that regular (ie periodic/quasi-periodic) motion dominates for the majority of initial conditions. However, the perturbation has caused the creation of alternating elliptic and hyperbolic fixed points. The majority of the phase space in this instance is made up of regular motion around the elliptic fixed points.

For $\epsilon=0.01$, there is now a large chaotic sea, which has formed around the hyperbolic fixed points. As the size of the perturbation (or the system energy) increases, the proportion of initial conditions that exhibit chaos increases.

In the final section, the motion is fully chaotic. Each of the paths ergodically explores the energy cross section. 
\begin{figure}
\centering
\includegraphics[scale=0.28]{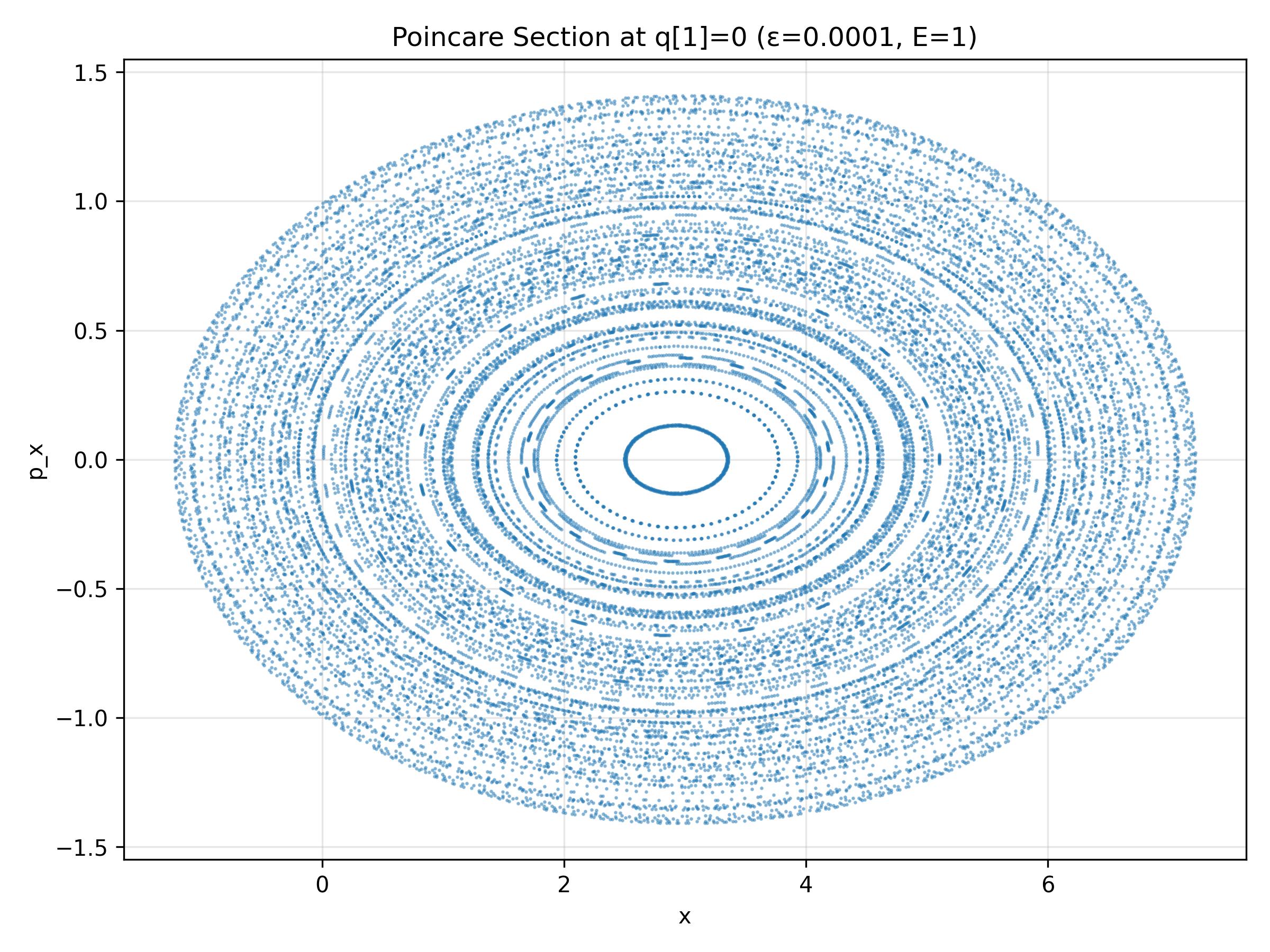}
\includegraphics[scale=0.28]{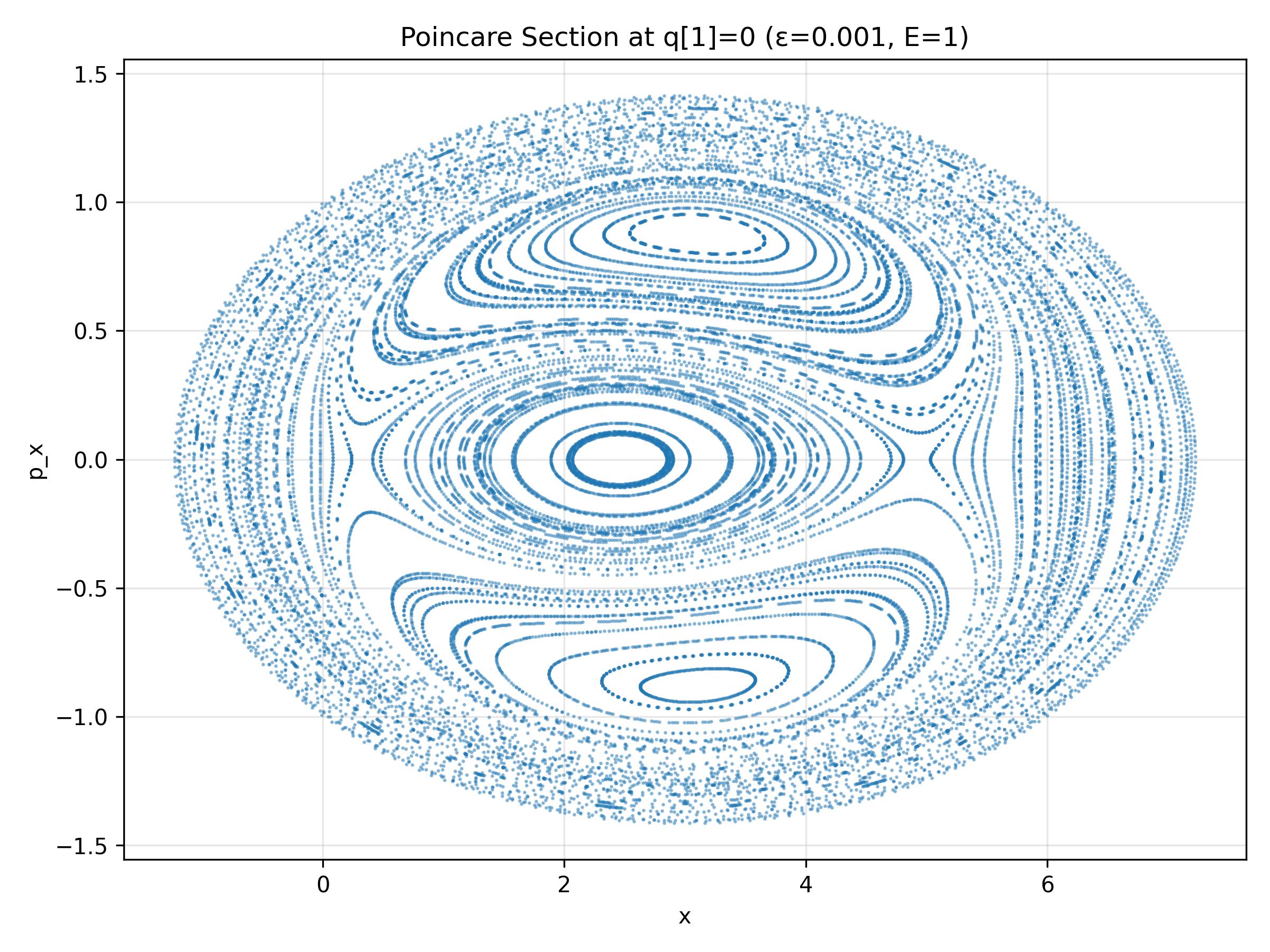}
\includegraphics[scale=0.28]{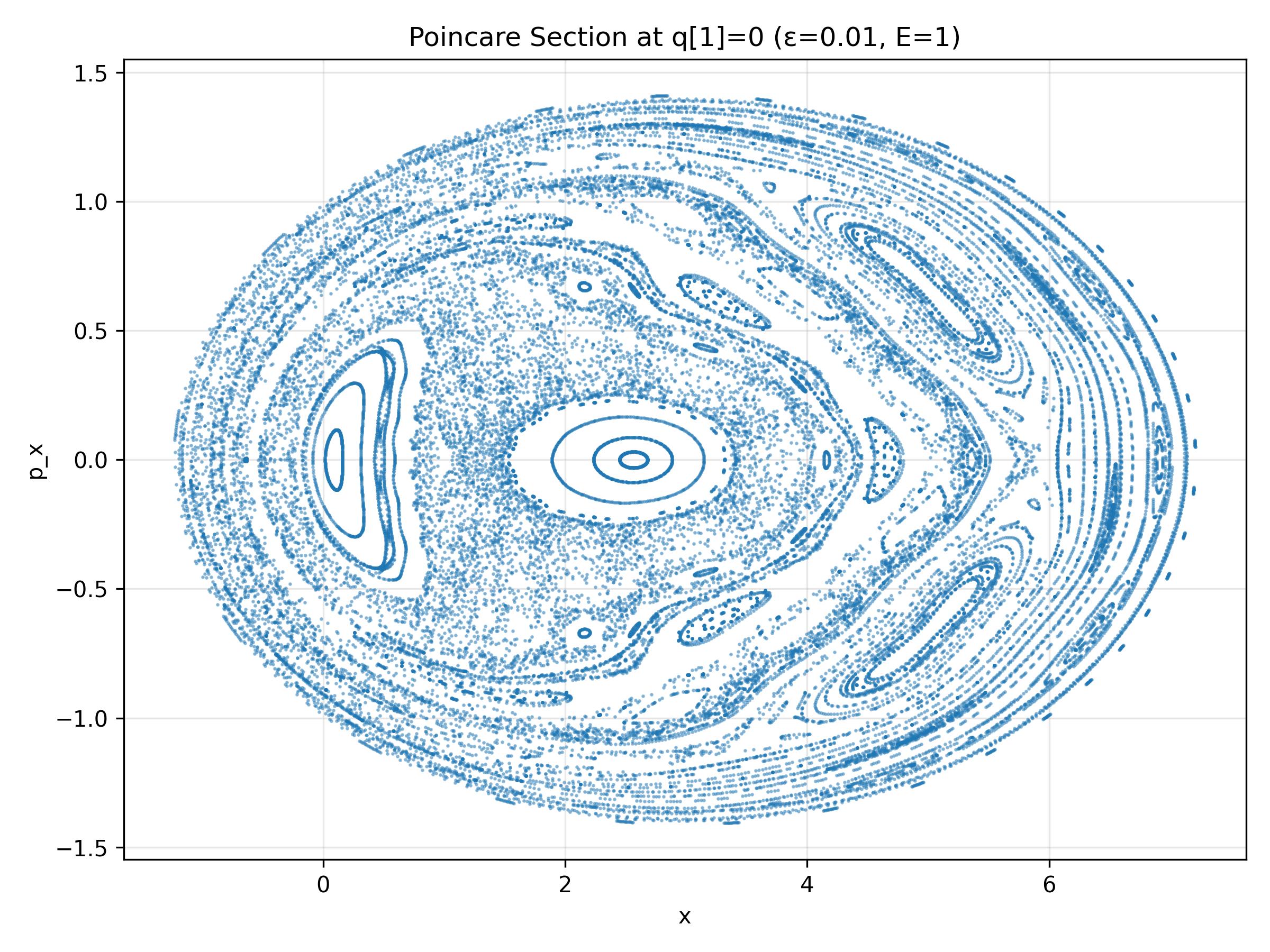}
\includegraphics[scale=0.28]{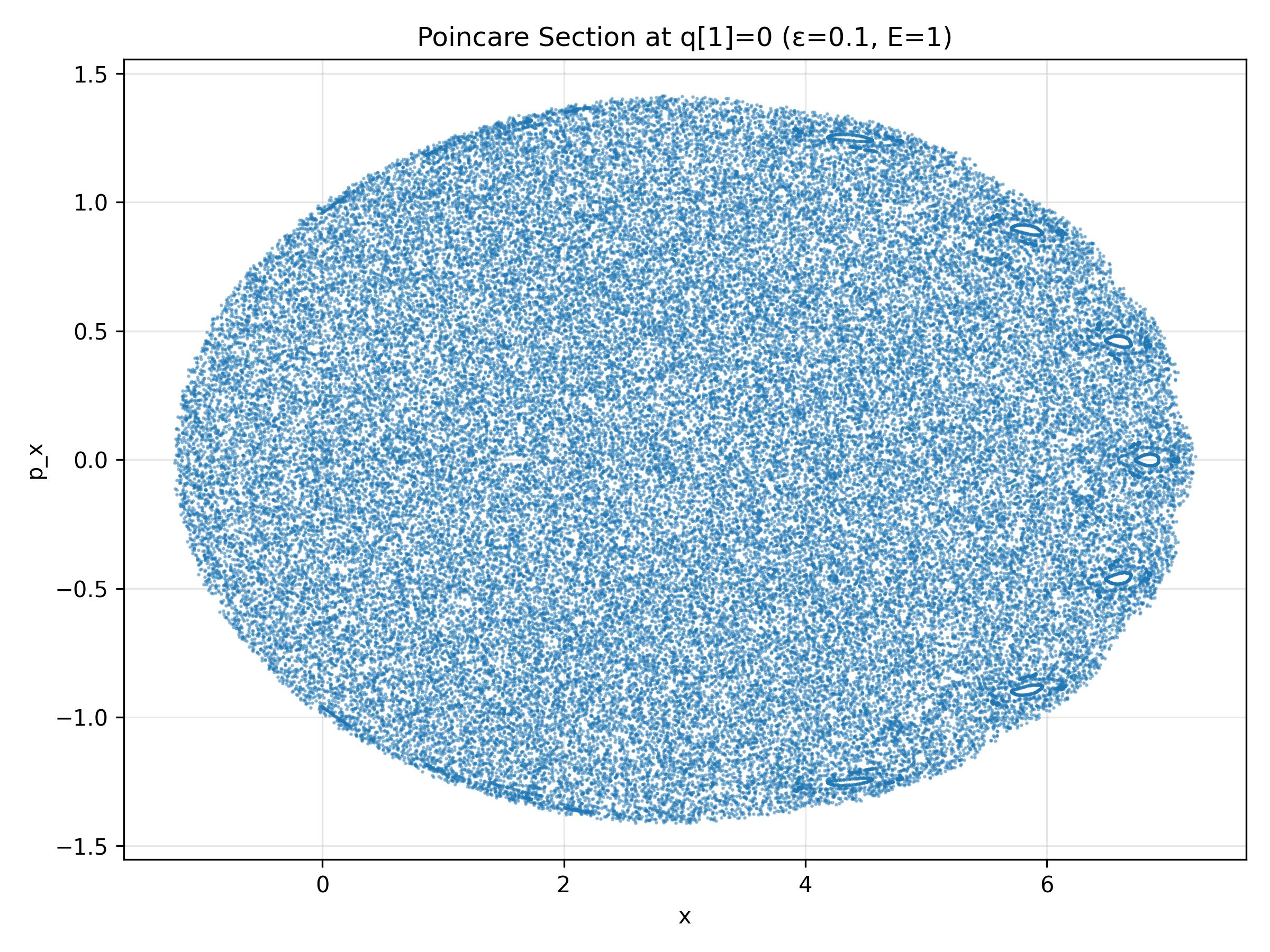}
\caption{Poincar\'{e} sections for $v=0$, $\dot{v}>0$ at energy=1$\pm 0.01$. We show the section with $\epsilon=0.0001$, $0.001$, $0.01$ and $0.1$}\label{PC_E_1}
\end{figure}
\begin{figure}
\includegraphics[scale=0.2]{poincare_eps_0.001_E_1.jpg}
\includegraphics[scale=0.2]{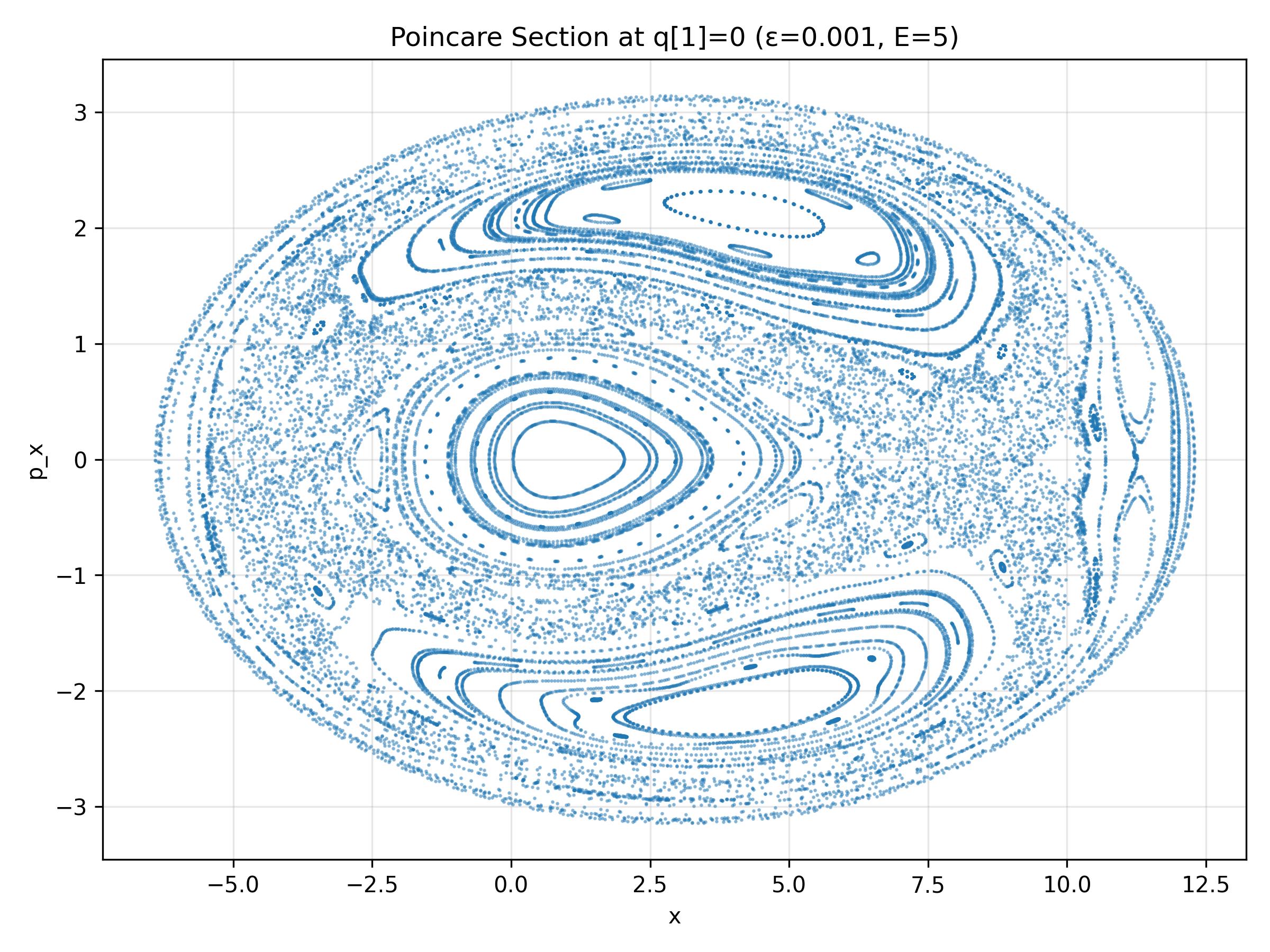}
\includegraphics[scale=0.2]{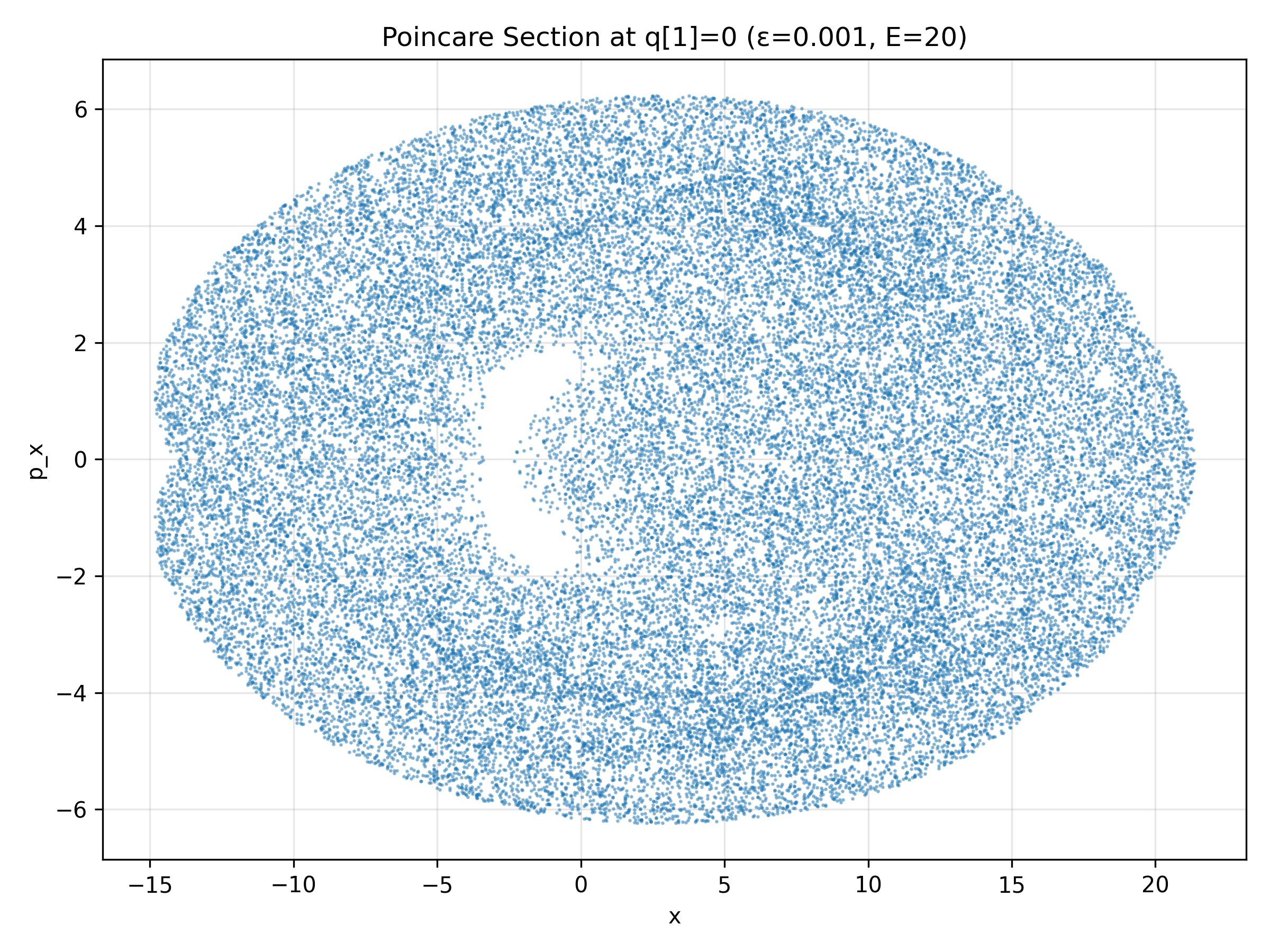}
\caption{Poincar\'{e} sections for $v=0$, $\dot{v}>0$ at energy=1$\pm 0.01$, 5$\pm 0.01$, and 20$\pm 0.01$. $\epsilon=0.001$.}\label{PC_eps_0p001}
\end{figure}
Figure \ref{PC_eps_0p001} shows a similar picture. Here the value for  $\epsilon$ is fixed at 0.001, and the energy target is set to 1, 5, and 20. With $\epsilon=0.001$ and the energy target at 20, the Chaotic motion is not fully ergodic, with a noticeable Cantori in the centre of the Poincar\'{e} section.
\subsection{Lyapunov Exponents and the Kolmogorov-Sinai Entropy}\label{L_KS}
The key phenomenon that is under investigation is the unpredictability of the changes in financial market prices. Put another way, when we observe new price movements we are accruing information. One way of measuring the rate at which new information is generated by a process is through the Kolmogorov-Sinai entropy:
\begin{definition}\label{KS_ent}
Let $P_{i_0,i_1,\dots,i_n}$ represent the probability that a variable is in hypercube $i_0$ to begin with, in $i_1$ after 1 step, and so on. Then the Kolmogorov-Sinai entropy is defined by the supremum over all partitions of:
\begin{align*}
H_{KS}&=\lim_{n\rightarrow\infty}\sum_{n=1}^{N-1}K_{n+1}-K_n\text{, }K_n=-\sum_{i_0,\dots,i_n}P_{i_0,i_1,\dots,i_n}\log(P_{i_0,i_1,\dots,i_n})
\end{align*}
\end{definition}
Note the following:
\begin{itemize}
\item For a process exhibiting regular motion (ie periodic, or quasi-periodic) the Kolmogorov-Sinai entropy is zero. The evolution of the process is fully predictable, and there is no new information.
\item For a non-deterministic Markov process continuous in state and time, that is non deterministic, for example a Wiener process, the Kolmogorov-Sinai entropy is infinite (see \cite{OrnsteinShields}). The fractal nature of the process means that finer partitions will always reveal more information.
\item For chaotic motion, such as that seen here, the Kolmogorov-Sinai entropy is positive, but finite. Eventually, increasing the measurement frequency (via a more granular partition) will not lead to any more information.
\end{itemize}
In fact, from Pesin's formula (see \cite{Pesin}) $H_{KS}$ can be estimated from the sum of all positive Lyapunov exponents for the process:
\begin{align*}
H_{KS}&=\sum_{\lambda_i>0}\lambda_i
\end{align*}
In figure \ref{KS_lyap} we show an estimate of the Kolmogorov-Sinai entropy, calculated using the sum of positive Lyapunov exponents, for energy of $5$, and $\epsilon$ values ranging from 0.001 to 0.1. These have been estimated using 5 sample paths (5 for each value of $\epsilon$) with an energy tolerance of $\pm 0.01$. The error bars show the range of values obtained from the sample paths used.
\begin{figure}
\includegraphics[scale=0.4]{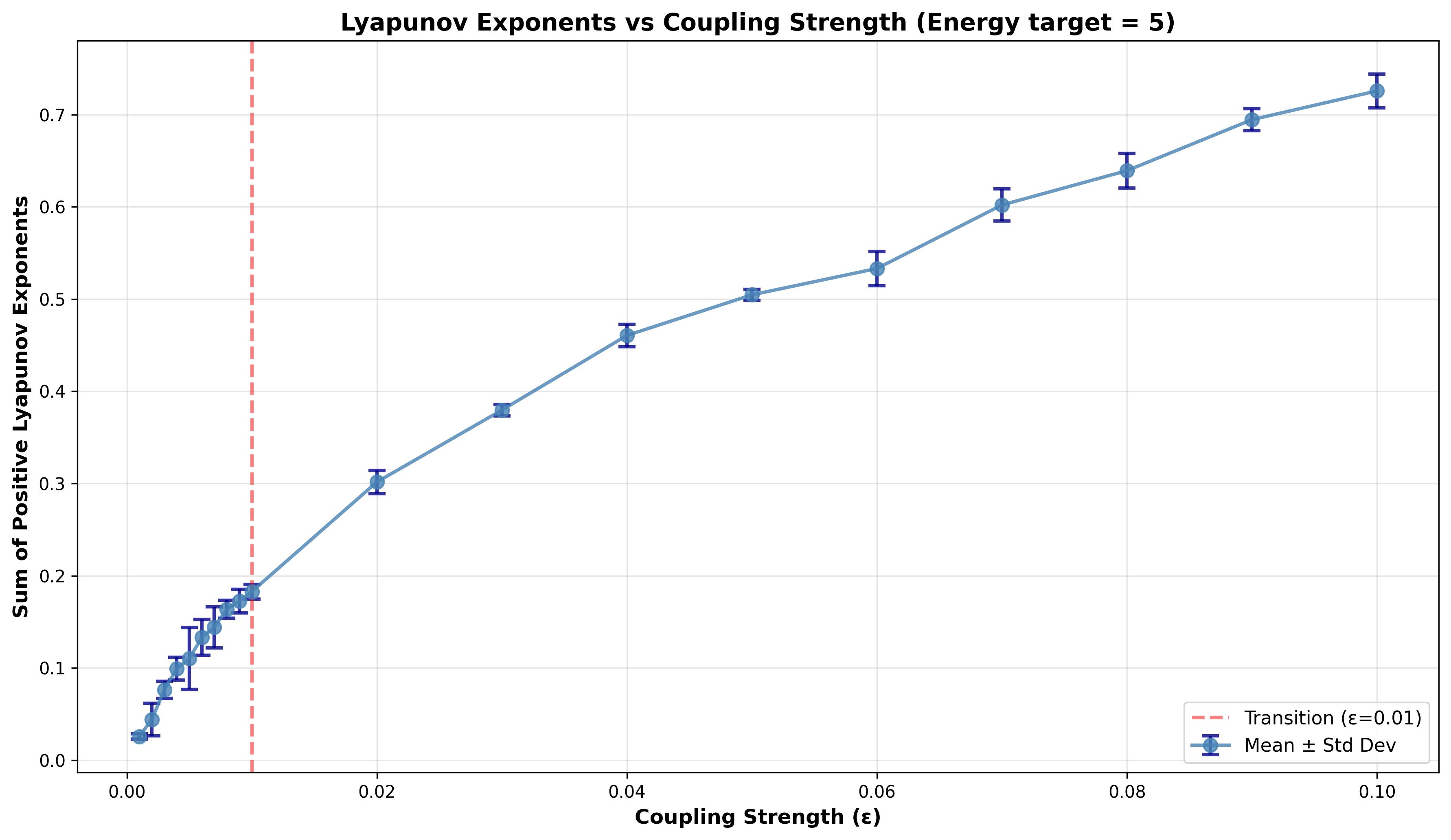}
\caption{The chart shows the sum of positive Lyapunov exponents, calculated using paths with energy target $5$, and $\epsilon$ ranging from $0.001$ up to $0.1$}\label{KS_lyap}
\end{figure}.
\subsection{Emergence of Unpredictability}
The original question asked was in relation to whether one can generate unpredictable prices from a simple closed (Hamiltonian) system. As discussed in section \ref{L_KS}, a chaotic process can never be truly `random' in the sense of the type of diffusion process commonly applied in finance. The Kolmogorov-Sinai entropy will always remain finite. To emphasise this point, in figure \ref{X_traj} we show an example process for the price $X$, with energy $10.6$, and $\epsilon=0.1$. It is clear from this trajectory that the process, whilst chaotic, is not truly random. The question therefore arises as to whether such a process can ever truly become indistinguishable from processes, such as a financial time-series, that are modelled using random noise and stochastic processes.

The Lyapunov time: $T_{\lambda}=1/\lambda_{max}$ represents an estimate for how long before tiny fluctuations in the measurement of initial conditions, explode into significant differences. If we assume that the frequency with which we can make measurements of a random process is limited, and as such the typical period between observations is $T_{obs}$, then as $\lambda_{max}$ grows, once $T_{\lambda}<<T_{obs}$, the chaotic process will become indistinguishable from random noise. In fact, the true `price' actually quoted on data aggregation platforms is estimated periodically updated based on a continuous stream of underlying exchange orders. With this in mind a limitation in the frequency with which observations can be made seems reasonable.

We show the process periodically sampled exactly every 100 steps in the bottom left chart of \ref{X_traj}. This process now looks more convincingly `random'. A histogram for the full set of $x$ values ($E=10.6,\epsilon=0.1$) is shown in the bottom left.
\begin{figure}
\centering
\includegraphics[scale=0.42]{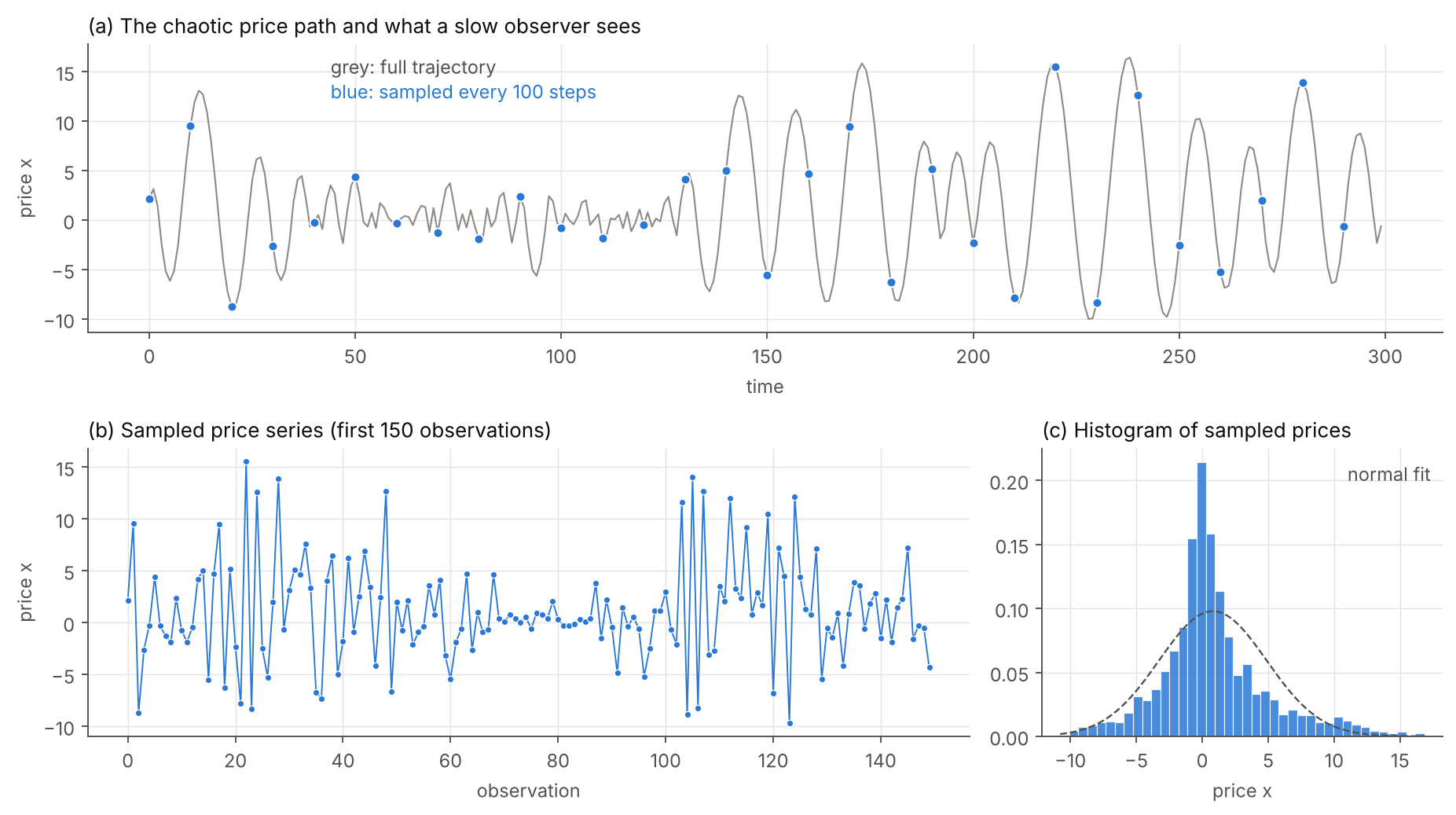}
\caption{The top chart shows 3000 steps of an example price process $X$, where the system energy is 10.6, and $\epsilon=0.1$. The bottom left chart shows the same process sampled every 100 steps. Finally, the histogram on the bottom right shows the summary of 2000 prices sampled every 100 steps.}\label{X_traj}
\end{figure}
\subsubsection{Measuring Unpredictability:}
In order to establish the notion of unpredictability we carry out the following test:
\begin{enumerate}
\item[Step1)] For $m=2,3,4$ form an m-vector of the latest prices:$x_{t-m+1},\dots,x_t$. We seek to predict $x_{t+1}$.
\item[Step2)] We split the simulated data set into two halves. The historical data used for forecasting, and the test set.
\item[Step3)] We find the 8 historical m-vectors with the minimum euclidean distance from the current m-vector. The prediction is the average of the next price in the sequence for a vector in the historical set of vectors used for forecasting.
\item[Step4)] We measure the root mean square error as a fraction of the standard deviation. We report the lowest error from the 3 predictions ($m=2,3,4$).
\item[Step5)] We run this test for 30 paths with random initial conditions fit to the initial energy=10.6. The dashed line shows the median error for the 30 paths, where the full series of prices has been randomly shuffled (ie usable history is destroyed).
\end{enumerate}
The results for $\epsilon=0.1$ are shown in fig \ref{X_traj}, against the base case path with $\epsilon=1e-4$. A forecast error equal to 1 std deviation shows that the forecast is no better than guessing using the series average. In each case we show the median value, with the shaded area representing the 10\% to 90\% values. The results show a value of 0.97 with a sampling interval of 10 Lyapunov times (sampling every 215 steps). By the time the sampling interval reaches 500 steps/23 Lyapunov times, the series shows the same forecast error as the randomly shuffled comparison. The results for $\epsilon=1e-4$ remain predictable, even at this level.
\begin{figure}
\includegraphics[scale=0.65]{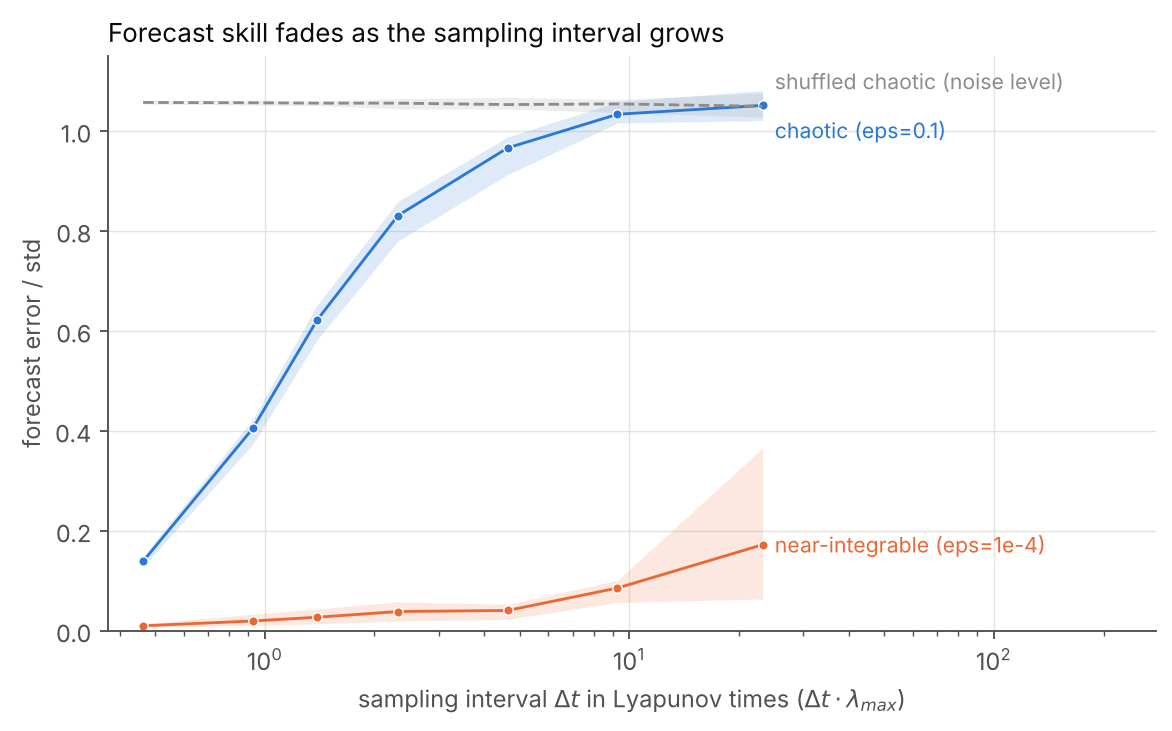}
\caption{The chart shows the forecast error measured in units of std deviation for the $\epsilon=0.1$ case, and a base case using $\epsilon=1e-4$.}\label{fig5}
\end{figure}
\subsubsection{Hypothesis Tesing vs independent and identically distributed variables}\label{hype_test}
Table \ref{fig2} shows the results for the 4 hypothesis tests described in table \ref{tab:iid_tests}. In each case, the null hypothesis being that the path is generated by independent and identically distributed variables. The fifth row rejects the hypothesis if any of the p-values falls below $1.25\%$. 
\begin{figure}[H]
\includegraphics[scale=0.33]{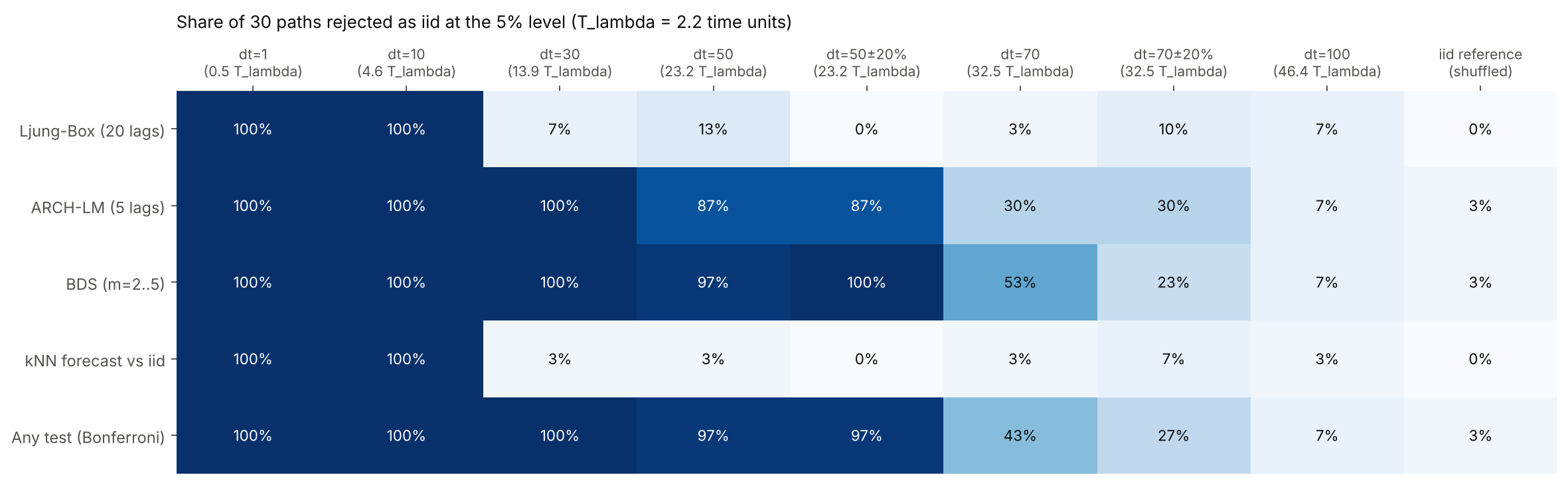}
\caption{The table shows the proportion of 30 test paths ($E=10.6$, $\epsilon=0.1$) rejecting the null hypothesis that the paths are generated by independent and identically distributed returns. We the show the results for the hypothesis tests described in table \ref{tab:iid_tests} for various sampling intervals.}\label{fig2}
\end{figure}
\begin{itemize}
\item The Ljung-Box test shows that for sampling intervals equal or greater than 300 steps, the linear autocorrelation has disappeared. The kNN forecast test at this sampling interval shows that the results predictability matches independent identically distributed variables.
\item The ARCH-LM test shows that the simulated paths show a higher degree of volatility clustering, right up to the 700 step sampling interval, than would be expected from independent and identically distributed data.
\item The BDS test looks at steps of $m$ sequences of prices being close in the sense of maximum Chebyshev distance. Note that previous authors have noted that real financial data often fail the BDS test (see for example \cite{Brock}) and that this is often evidence of volatility clustering. Ie large financial returns are more likely to be seen adjacent to other large returns than would be the case for independent and identically distributed returns.
\end{itemize}
\begin{table}[ht]
\centering
\small
\renewcommand{\arraystretch}{1.25}
\begin{tabular}{@{}>{\raggedright\arraybackslash}p{2cm}
                >{\raggedright\arraybackslash}p{2.4cm}
                >{\raggedright\arraybackslash}p{2.6cm}
                >{\raggedright\arraybackslash}p{4.6cm}@{}}
\toprule
\textbf{Test} & \textbf{Applied to} & \textbf{Detects} & \textbf{Rejects the iid null when\ldots} \\
\midrule
Ljung--Box (20 lags) & Sampled prices
  & Linear autocorrelation
  & the sum of squared autocorrelations over lags 1--20 is larger than expected for iid data. \\
ARCH-LM (5 lags) & AR residuals
  & Volatility clustering
  & squared residuals are significantly predicted by their own 5 previous values. \\
BDS ($m=2,\dots,5$) & AR residuals
  & Any remaining dependence, linear or nonlinear
  & length-$m$ stretches of the series are close together more (or less) often than iid data implies; $p$-values over $m$ are Bonferroni-combined. \\
kNN forecast & Sampled prices
  & Predictability from the price history
  & the out-of-sample nearest-neighbour forecast error beats almost all of 39 shuffled copies of the series. \\
\midrule
Any test & ---
  & At least one of the above
  & any of the four $p$-values is below $0.05/4$ (Bonferroni). \\
\bottomrule
\end{tabular}
\caption{The iid tests in Figure~2. Each is applied to every simulated path, and the figure reports the share of paths rejected at the 5\% level. Residuals are from an AR($p$) fit to the sampled price, with $p$ chosen by BIC.}
\label{tab:iid_tests}
\end{table}
\subsubsection{Fitted GARCH Process:}
In this section, we fit a GARCH process:
\begin{align}
  x_t &= \mu + \sum_{k=1}^{p} \phi_k \,(x_{t-k} - \mu) + \varepsilon_t, \label{eq:ar}\\[2pt]
  \varepsilon_t &= \sigma_t z_t, \qquad z_t \;\text{i.i.d.},\quad \mathbb{E}[z_t]=0,\ \mathrm{Var}[z_t]=1,\nonumber\\[2pt]
  \sigma_t^2 &= \omega + \alpha\, \varepsilon_{t-1}^{2} + \beta\, \sigma_{t-1}^{2},\nonumber
\end{align}
with $\omega>0$, $\alpha,\beta\ge 0$ and $\alpha+\beta<1$.
Here $\mu$ is the sample mean of the price, the order $p\in\{0,\dots,5\}$ is chosen by the Bayesian information criterion (BIC) and the
$\phi_k$ by least squares. $(\omega,\alpha,\beta)$ are estimated by Gaussian quasi-maximum likelihood,
\begin{equation}
  (\hat\omega,\hat\alpha,\hat\beta) = \arg\min_{\omega,\alpha,\beta}\;
  \frac12 \sum_t \left[ \ln \sigma_t^2 + \frac{\varepsilon_t^2}{\sigma_t^2} \right],
\end{equation}
which does not require $z_t$ to be Gaussian. In the simulated baseline, $z_t$ is drawn by resampling the
standardised residuals $\hat z_t = \hat\varepsilon_t/\hat\sigma_t$ of the fit to the data. The resulting fitted parameters are shown in table \ref{garch}.
\begin{figure}[H]
\includegraphics[scale=0.44]{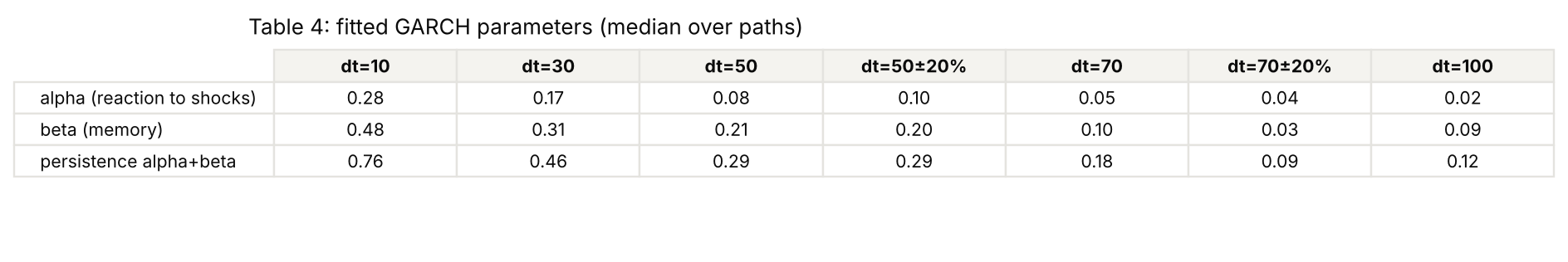}
\caption{The table shows the resulting $\alpha$ and $\beta$ parameters for the fitted GARCH processes.}\label{garch}
\end{figure}
The results are summarised in table \ref{garch_hyp} and figure \ref{garch_volmem}. In table \ref{garch_hyp} a) (on the left), the same hypothesis testing from section \ref{hype_test} are shown. The non-bracketed numbers show the results from the simulated paths. The bracketed numbers show the results from the fitted GARCH process. In table b) (on the right) we apply these tests to the residuals generated by the GARCH fitting from the simulated process. We simulate 30 GARCH processes using the fitted GARCH parameters, and the bracketed numbers show the rejection rate based residuals from the simulated GARCH processes. This gives a base case, where the GARCH process matches exactly the underlying dynamics.

At sampling intervals of about $\Delta t=50$ (roughly 23 Lyapunov times), the standardised residuals of the fitted GARCH model are statistically indistinguishable from iid noise under the four tests considered. The rejection rate matches the rate expected for data generated by a GARCH process itself. For this sampling interval the deterministic dynamics are therefore statistically consistent with a stochastic-volatility description, although the BDS test still detects somewhat more structure in the sampled prices than in simulated GARCH series. At shorter sampling intervals the process is distinguishable from GARCH. For $\Delta t=100$, both the simulated paths, based on these tests, becomes statistically indistinguishable from independent and identically distributed noise.

Figure \ref{garch_volmem} shows the autocorrelation of the squared residuals versus the fitted GARCH processes, for a variety of sampling intervals.
\begin{figure}
\includegraphics[scale=0.33]{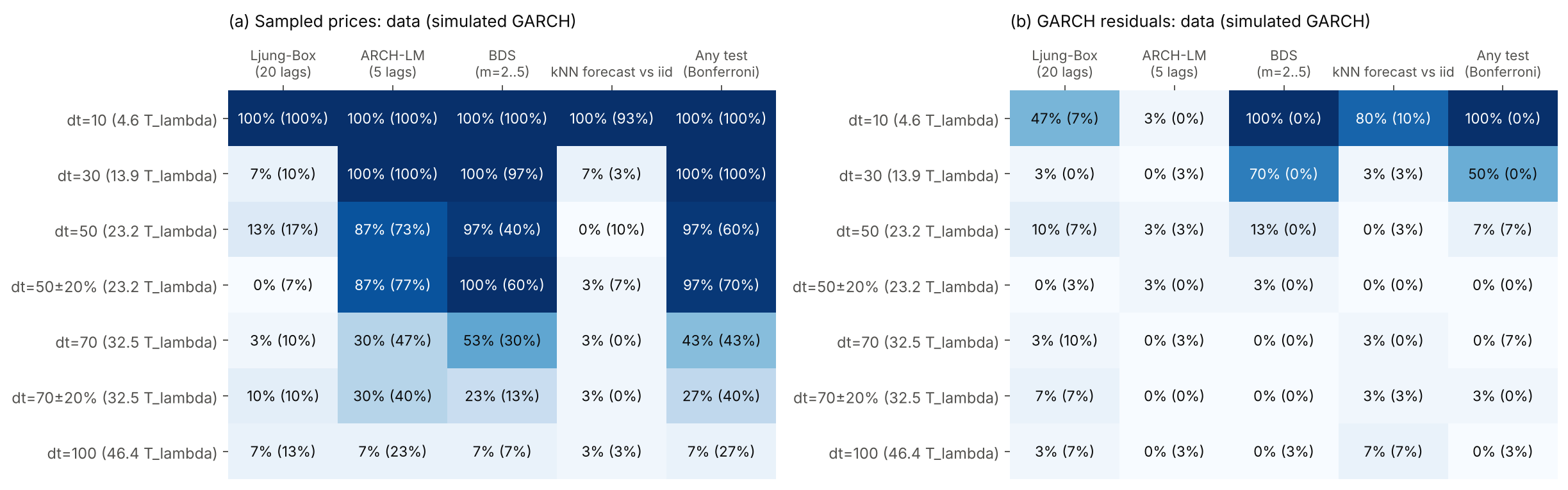}
\caption{The table shows the same hypothesis testing from section \ref{hype_test}, applied to the GARCH process, and the residuals vs the simulated paths.}\label{garch_hyp}
\end{figure}
\begin{figure}
\includegraphics[scale=0.38]{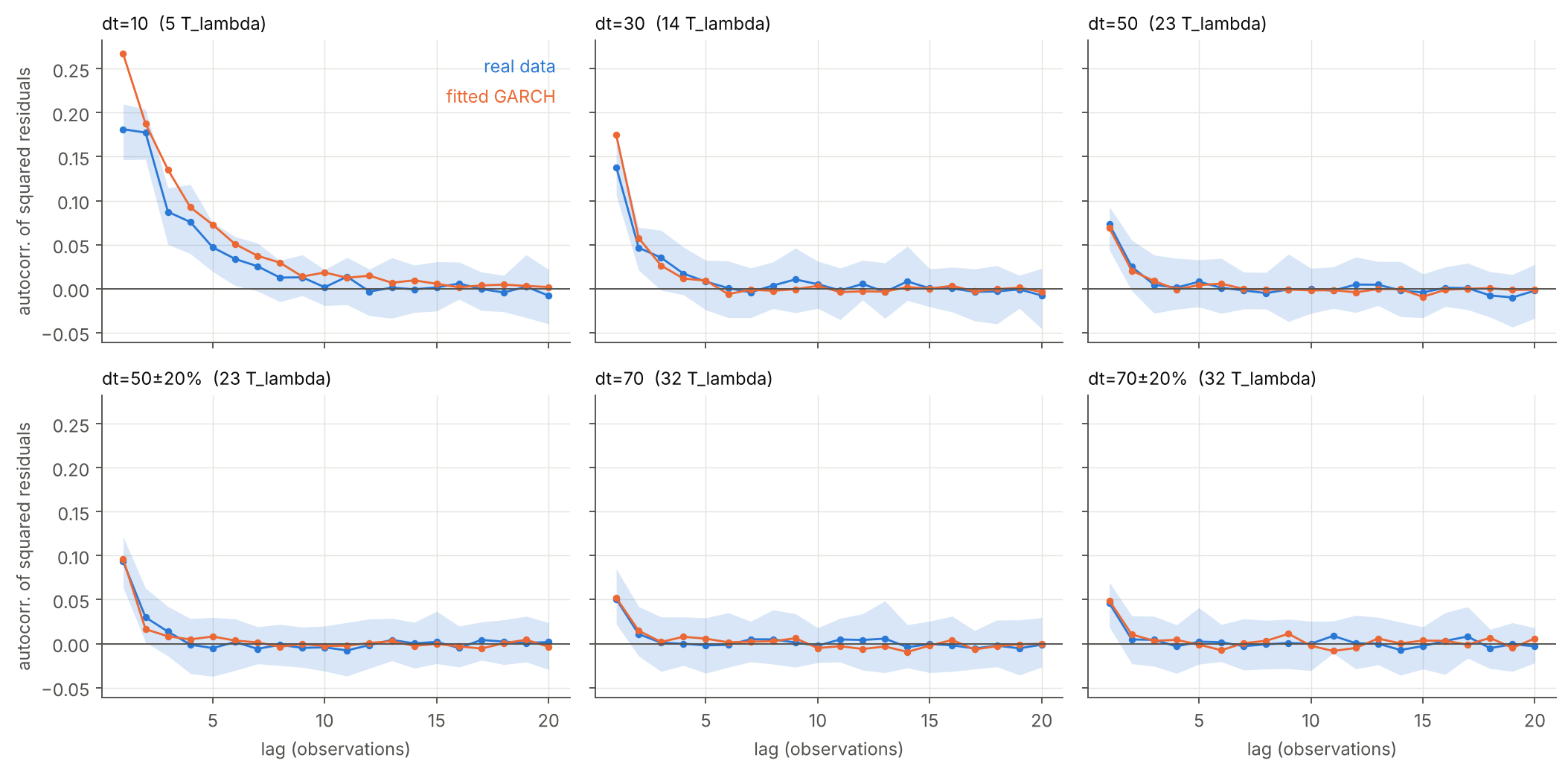}
\caption{The table shows the same hypothesis testing from section \ref{hype_test}, applied to the GARCH process, and the residuals vs the simulated paths.}\label{garch_volmem}
\end{figure}
\section{Conclusion:}
The analysis shows that, whilst external shocks and interaction with the external economic environment undoubtedly impact the price movements of assets traded on financial exchanges, random noise is not required to generate unpredictable price changes. By encoding relatively simple market forces that impact prices into a Hamiltonian, one sees price processes that are not easily distinguishable from time-series generated using random noise. Furthermore, by encoding market forces into a Hamiltonian, one can model real financial artefacts into the model dynamics in a natural way. For example the simple models illustrated in this article, whilst not necessarily being realistic for every day use, enable factors such as the impact of market maker inventory and risk appetite, as well as market breakdown, into modelling the unpredictable changes in financial markets.\newline

We find that where the market makers have a high risk appetite/low risk aversion, they can always absorb sufficient inventory to ensure the price returns to equilibrium in an orderly fashion. On the other hand, a low risk appetite/high risk aversion leads to unpredictable price changes, a depressed equilibrium price, and a degree of volatility clustering.
\section*{Declaration of Conflicting Interests}
The author declares no potential conflicts of interest with respect to the research, authorship, and/or publication of this article.

\end{document}